\documentclass{article}
\usepackage{amsmath}
\usepackage{amssymb}
\usepackage{amsfonts}
\usepackage{a4wide}
\usepackage{amsthm}
\usepackage{tikz}
\usepackage{float}
\usepackage{subcaption}
\usepackage{enumerate}

\usetikzlibrary{arrows}

\newcounter{NN}
\newtheorem{proposition}[NN]{Proposition}
\newtheorem{theorem}[NN]{Theorem}
\newtheorem{corollary}[NN]{Corollary}

\newtheorem{definition}[NN]{Definition}
\newtheorem{example}[NN]{Example}

\def\adj#1{\mathrm{adj}(#1)}

\def\A{\mathbf{A}}
\def\B{\mathbf{B}}
\def\Z{\mathbf{Z}}
\def\C{\mathbf{C}}
\def\I{\mathbf{1}}
\def\v{\mathbf{v}}
\def\w{\mathbf{w}}
\def\x{\mathbf{x}}
\def\P{\mathbf{P}}
\begin{document}

\title{Trees and superintegrable Lotka-Volterra families}
\author{
Peter H.~van der Kamp, G.R.W. Quispel and David I. McLaren\\[2mm]
Department of Mathematical and Physical Sciences,\\
La Trobe University, Victoria 3086, Australia.\\[2mm]
Corresponding author: Peter H.~van der Kamp\\
Email: P.vanderKamp@LaTrobe.edu.au\\[7mm]
}

\maketitle

\begin{abstract}
To any tree on $n$ vertices we associate an $n$-dimensional Lotka-Volterra system with $3n-2$ parameters and, for generic values of the parameters, prove it is superintegrable, i.e. it admits $n-1$ functionally independent integrals. We also show how each system can be reduced to an ($n-1$)-dimensional system which is superintegrable and solvable by quadratures.
\end{abstract}

\section{Introduction}
This paper is concerned with {\em trees}, {\em Lotka-Volterra systems}, {\em superintegrability}, and {\em Darboux Polynomials}. We first define these notions and try to indicate why the reader perhaps might care.

Lotka-Volterra (LV) systems are quadratic ordinary differential equations (ODEs) of the form
\begin{equation}\label{LVsys}
\frac{dx_i}{dt} = x_i\left(b_i + \sum_{j=1}^{n}A_{ij}x_j\right),\quad i=1,\dots,n \quad ,
\end{equation}
where the vector $\mathbf{b}$ and the matrix $\A$ do not depend on $\mathbf{x}$ or $t$. LV equations were introduced to describe predator-prey systems \cite{Lot,Vol}, and have more recently been used e.g. to study infectious disease \cite{HolPic}, and competition between species \cite{May}. The papers \cite{Bog, CD, Dam, Kolm, KQV, Maier, Pla, Smale, ProcA}, and many others, give some indication of the scientific interest generated by Lotka-Volterra equations, their properties, and their generalisations.

A Darboux Polynomial (DP) for an ODE $\frac{dx}{dt}=f(x)$ is a polynomial $P(x)$ for which there exists another polynomial $C(x)$ (called the cofactor) such that
\begin{equation}\label{DPdef}
\frac{dP}{dt} :=  \nabla P \cdot f = CP
\end{equation}
Comparing eqns (\ref{LVsys}) and (\ref{DPdef}) we see immediately that each LV system has at least $n$ DPs (i.e. the $x_i, \, i=1,\dots,n$). From the theory of Darboux Polynomials \cite{Darboux,Goriely} it follows that if a homogeneous quadratic ODE possesses $n+1$ DPs, it must have a factorisable integral (whose factors are powers of DPs). One might say that LV systems are on the verge of possessing one or more first integrals.\footnote{We refer to e.g. \cite[Theorem 3.1]{Zha} for a general result on the relationship between the existence of a sufficient number of Darboux polynomials and the existence of a first integral obtained from them.} In \cite{QTMK} we constructed families of homogeneous (i.e. with $\mathbf{b}=\mathbf{0}$) $n$-dimensional LV systems that possess $n-1$ additional DPs which give rise to $n-1$ first integrals. In this paper we prove that these integrals are functionally independent, and hence that these systems are superintegrable. A complementary study was carried out in \cite{LRR}, where conditions were given for systems \eqref{LVsys} with one additional DP (invariant hyperplane) to be superintegrable.\footnote{The condition is that the ranks of certain matrices are low \cite[Theorems 2,5]{LRR}. The statement in \cite[Theorem 3, condition (12)]{LRR} should read \cite{Lli}: Moreover, if rank$(B_1)=$rank$(B_2)=$rank$(B_3)=3\leq N$, then differential system (3) is completely integrable.} We note that while in this paper we consider homogeneous system only, the results carry over to the case where $\mathbf{b}=b\I$ is a constant vector, as such an LV-system can be transformed into one with $\mathbf{b}=\mathbf{0}$ \cite{Maier}. Homogeneous DPs are preserved under the inverse transformation, however, they can give rise to time-dependent integrals.
 
A tree is a connected graph which does not contain a cycle. In the next section, we start from an undirected tree on $n$ vertices with $n-1$ edges, and extend it to a weighted complete directed graph (digraph), which includes both loops and multiple edges. We then use the adjacency matrix of this digraph to define a corresponding LV system. In section 3, we establish some properties of the minors of the adjacency matrix, which are used in section 4 to derive $n-2$ three-point integrals, and to factorise the exponents in the integrals for the LV systems (which were given in \cite{QTMK}). In section 5, we introduce so-called edge variables, in terms of which the $n-2$ three-point integrals become two-point integrals, and we derive a related $(n-1)$-dimensional system which preserves the same integrals. This system is shown to be solvable by quadratures (Theorem \ref{nm1}). Finally, in section 6, we utilise the edge-variables and the 2-point integrals and to prove functional independence for $n-1$ integrals of the tree-system, i.e. that tree-systems are superintegrable (Theorem \ref{tfii}).

The number of non-isomorphic trees on $n$ vertices \cite[A000055]{OEIS} is quite large, asymptotically equal to $\beta\cdot\alpha^n\cdot n^{-5/2}$, where $\alpha\approx 2.955765, \beta\approx 0.5349496$ are Otter's tree-constants \cite[A051491]{OEIS}, \cite[A086308]{OEIS}, \cite{Otter}. We establish the existence of an equal number of superintegrable $n$-dimensional LV systems with $3n-2$ parameters.

\section{From tree to Lotka-Volterra system}
Let $T$ be a tree on $n$ vertices, $1,2,\ldots,n$. Edges $e_{1}, \ldots,e_{n-1}$ are represented by ordered pairs of integers,  i.e. we have $e_i=(u_i,v_i)\in \mathbb{N}^2$ with $u_i<v_i$. Because of the ordering the tree becomes a directed tree.

\begin{definition}
We associate a weighted digraph 
$G$ to $T$ as follows. For all $1\leq i\leq n$, we add a loop with weight $a_i$ to vertex $i$. We add weights $b_i$ to edges $e_i=(u_i,v_i)$, and add reciprocal edges $e'_i=(v_i,u_i)$ which carry weights $c_i$. The adjacency matrix of $G$ is the $n\times n$ matrix $\A$ with $A_{i,i}=a_i$, and for each edge $e_i=(u_i,v_i)$ of $T$: $A_{u_i,v_i}=b_i$ and $A_{v_i,u_i}=c_i$, and with all other entries 0.
\end{definition}

\begin{figure}[h]
\centering
\begin{subfigure}{0.4\textwidth}
\begin{tikzpicture}[scale=.85]
\tikzset{vertex/.style = {shape=circle,draw,minimum size=1.5em}}
\tikzset{edge/.style = {->}}				
\node[shape=circle,draw=black,line width=0.5mm] (1) at (0,4) {1};
\node[shape=circle,draw=black,line width=0.5mm] (2) at (2,4) {2};
\node[shape=circle,draw=black,line width=0.5mm] (3) at (4,4) {3};
\node[shape=circle,draw=black,line width=0.5mm] (4) at (6,4) {4};
\node[shape=circle,draw=black,line width=0.5mm] (5) at (2,0) {5};
\node[shape=circle,draw=black,line width=0.5mm] (6) at (2,2) {6};
\node[shape=circle,draw=black,line width=0.5mm] (7) at (4,6) {7};
\node[shape=circle,draw=black,line width=0.5mm] (8) at (2,6) {8};
\draw[edge] (1) to node[above] {$e_1$} (2);
\draw[edge] (2) to node[above] {$e_3$} (3);
\draw[edge] (3) to node[above] {$e_5$} (4);
\draw[edge] (2) to node[left] {$e_3$} (6);
\draw[edge] (5) to node[left] {$e_7$} (6);
\draw[edge] (2) to node[left] {$e_4$} (8);
\draw[edge] (3) to node[left] {$e_6$} (7);
\end{tikzpicture}
\caption{A (directed) tree on 8 vertices with labeled edges.}
\end{subfigure}
\qquad
\begin{subfigure}{0.4\textwidth}
\begin{tikzpicture}[scale=.75]
\tikzset{vertex/.style = {shape=circle,draw,minimum size=1.5em}}
\tikzset{edge/.style = {->,> = latex'}}				
\node[shape=circle,draw=black,line width=0.5mm] (1) at (0,4) {1};
\node[shape=circle,draw=black,line width=0.5mm] (2) at (2,4) {2};
\node[shape=circle,draw=black,line width=0.5mm] (3) at (4,4) {3};
\node[shape=circle,draw=black,line width=0.5mm] (4) at (6,4) {4};
\node[shape=circle,draw=black,line width=0.5mm] (5) at (2,0) {5};
\node[shape=circle,draw=black,line width=0.5mm] (6) at (2,2) {6};
\node[shape=circle,draw=black,line width=0.5mm] (7) at (4,6) {7};
\node[shape=circle,draw=black,line width=0.5mm] (8) at (2,6) {8};
\draw[edge] (1.10) to node[above] {$b_1$} (2.170);
\draw[edge] (2.190) to node[below] {$c_1$} (1.350);	
\draw[edge] (2.10) to node[above] {$b_2$} (3.170);
\draw[edge] (3.190) to node[below] {$c_2$} (2.350);
\draw[edge] (2.260) to node[left] {$b_3$} (6.100);
\draw[edge] (6.80) to node[right] {$c_3$} (2.280);
\draw[edge] (2.100) to node[left] {$b_4$} (8.260);
\draw[edge] (8.280) to node[right] {$c_4$} (2.80);	
\draw[edge] (3.10) to node[above] {$b_5$} (4.170);
\draw[edge] (4.190) to node[below] {$c_5$} (3.350);
\draw[edge] (3.100) to node[left] {$b_6$} (7.260);
\draw[edge] (7.280) to node[right] {$c_6$} (3.80);	
\draw[edge] (5.100) to node[left] {$b_7$} (6.260);
\draw[edge] (6.280) to node[right] {$c_7$} (5.80);
\draw[edge] (1) to[in=160,out=200,looseness=5] node[midway,left] {$a_1$} (1);
\draw[edge] (2) to[in=115,out=155,looseness=5] node[midway,above left] {$a_2$} (2);
\draw[edge] (3) to[in=290,out=250,looseness=5] node[midway,below] {$a_3$} (3);
\draw[edge] (4) to[in=20,out=340,looseness=5] node[midway,right] {$a_4$} (4);
\draw[edge] (5) to[in=290,out=250,looseness=5] node[midway,below] {$a_5$} (5);
\draw[edge] (6) to[in=20,out=340,looseness=5] node[midway,right] {$a_6$} (6);
\draw[edge] (7) to[in=70,out=110,looseness=5] node[midway,above] {$a_7$} (7);
\draw[edge] (8) to[in=70,out=110,looseness=5] node[midway,above] {$a_8$} (8);	
\end{tikzpicture}
\caption{The associated weighted digraph.}
\end{subfigure}
\caption{Associating a weighted digraph to a tree. Added edges $(i,i)$ carry weights $a_i$, weights $b_i$ are connected to edges $e_i=(u_i,v_i)$ where $u_i<v_i$ and weights $c_i$ are connected to opposite edges $(v_i,u_i)$ with $v_i>u_i$.}
\label{assG}
\end{figure}
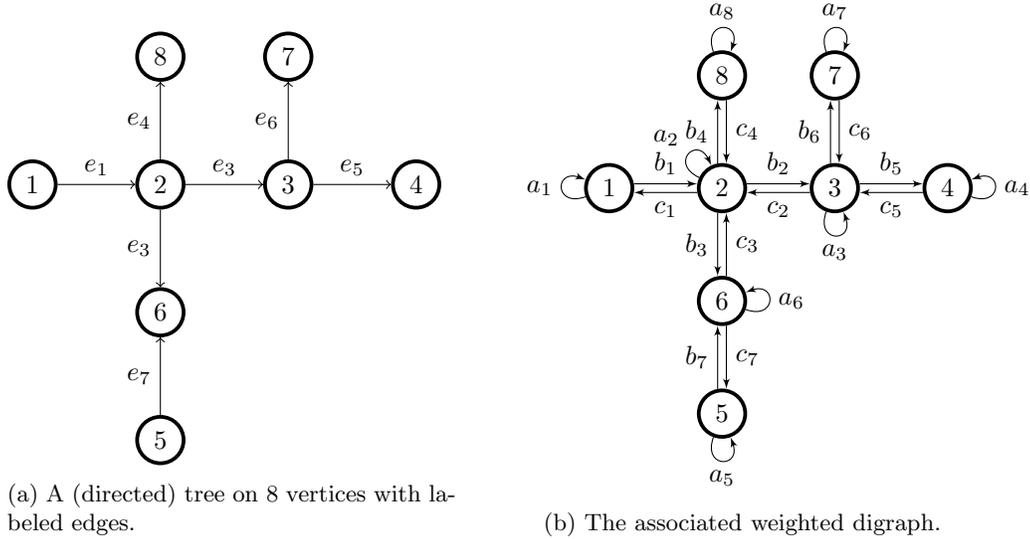

\begin{example} \label{exaTr}
To the tree on 8 vertices with edges
\[
e_1=(1,2),\ e_2=(2,3),\ e_3=(2,6),\ e_4=(2,8),\ e_5=(3,4),\ e_6=(3,7),\ e_7=(5,6),
\]
see Figure \ref{assG}(a), we associate the weighted digraph in Figure \ref{assG}(b).
The adjacency matrix of the associated weighted digraph is
\[
\begin{pmatrix}
a_{1} & b_{1} & 0 & 0 & 0 & 0 & 0 & 0
\\
 c_{1} & a_{2} & b_{2} & 0 & 0 & b_{3} & 0 & b_{4}
\\
 0 & c_{2} & a_{3} & b_{5} & 0 & 0 & b_{6} & 0
\\
 0 & 0 & c_{5} & a_{4} & 0 & 0 & 0 & 0
\\
 0 & 0 & 0 & 0 & a_{5} & b_{7} & 0 & 0
\\
 0 & c_{3} & 0 & 0 & c_{7} & a_{6} & 0 & 0
\\
 0 & 0 & c_{6} & 0 & 0 & 0 & a_{7} & 0
\\
 0 & c_{4} & 0 & 0 & 0 & 0 & 0 & a_{8}
\end{pmatrix}.
\]
\end{example}

\begin{definition} \label{dacdg}
We next associate a weighted complete digraph $\overline{G}$ to $T$. We add to the associated weighted digraph $G$ all edges $(x,z)$ which are not already in $G$. To assign their weight, we consider the shortest path in $T$ from $x$ to $z$ (this exists because $T$ is connected, and is uniquely given because $T$ does not contain loops). Let the penultimate vertex be $y$. Then there is an $i$ such that $(y,z)$ or $(z,y)$ is the $i$-th edge of $T$. If $y<z$ we assign to edge $(x,z)$ weight $b_i$, if $y>z$ we assign weight $c_i$ to edge $(x,z)$.
\end{definition}

In Figure \ref{aetassG} we add edges to the associated digraph from Figure \ref{assG}, which are directed to vertices 1 and 2. The complete digraph can be obtained by adding the remaining edges directed respectively to vertices $3,4,\ldots,8$.

\begin{figure}[h]
\begin{center}
\begin{tikzpicture}
\tikzset{vertex/.style = {shape=circle,draw,minimum size=1.5em}}
\tikzset{edge/.style = {->,> = latex'}}				
\node[shape=circle,draw=black,line width=0.5mm] (1) at (0,4) {1};
\node[shape=circle,draw=black,line width=0.5mm] (2) at (2,4) {2};
\node[shape=circle,draw=black,line width=0.5mm] (3) at (4,4) {3};
\node[shape=circle,draw=black,line width=0.5mm] (4) at (6,4) {4};
\node[shape=circle,draw=black,line width=0.5mm] (5) at (2,0) {5};
\node[shape=circle,draw=black,line width=0.5mm] (6) at (2,2) {6};
\node[shape=circle,draw=black,line width=0.5mm] (7) at (4,6) {7};
\node[shape=circle,draw=black,line width=0.5mm] (8) at (2,6) {8};
\draw[edge] (1.10) to node[above] {$b_1$} (2.170);
\draw[edge] (2.190) to node[below] {$c_1$} (1.350);	
\draw[edge] (2.10) to node[above] {$b_2$} (3.170);
\draw[edge] (3.190) to node[below] {$c_2$} (2.350);
\draw[edge] (2.260) to node[left] {$b_3$} (6.100);
\draw[edge] (6.80) to node[right] {$c_3$} (2.280);
\draw[edge] (2.100) to node[left] {$b_4$} (8.260);
\draw[edge] (8.280) to node[right] {$c_4$} (2.80);	
\draw[edge] (3.10) to node[above] {$b_5$} (4.170);
\draw[edge] (4.190) to node[below] {$c_5$} (3.350);
\draw[edge] (3.100) to node[left] {$b_6$} (7.260);
\draw[edge] (7.280) to node[right] {$c_6$} (3.80);	
\draw[edge] (5.100) to node[left] {$b_7$} (6.260);
\draw[edge] (6.280) to node[right] {$c_7$} (5.80);
\draw[edge] (1) to[in=160,out=200,looseness=5] node[midway,left] {$a_1$} (1);
\draw[edge] (2) to[in=115,out=155,looseness=5] node[midway,above left] {$a_2$} (2);
\draw[edge] (3) to[in=290,out=250,looseness=5] node[midway,below] {$a_3$} (3);
\draw[edge] (4) to[in=20,out=340,looseness=5] node[midway,right] {$a_4$} (4);
\draw[edge] (5) to[in=290,out=250,looseness=5] node[midway,below] {$a_5$} (5);
\draw[edge] (6) to[in=20,out=340,looseness=5] node[midway,right] {$a_6$} (6);
\draw[edge] (7) to[in=70,out=110,looseness=5] node[midway,above] {$a_7$} (7);
\draw[edge] (8) to[in=70,out=110,looseness=5] node[midway,above] {$a_8$} (8);
\tikzset{edge/.style = {->,densely dashed,> = latex'}}
\draw[edge] (8) to node[above left] {$c_1$} (1);
\draw[edge] (6) to node[below left] {$c_1$} (1);
\draw[edge] (5)[in=260,out=145,looseness=.8] to node[below left] {$c_1$} (1);
\draw[edge] (7)[in=80,out=120,looseness=1.5] to node[above left] {$c_1$} (1);
\draw[edge] (3)[in=330,out=210,looseness=1] to node[near start, below right] {$c_1$} (1);
\draw[edge] (4)[in=100,out=100,looseness=2] to node[above left] {$c_1$} (1);
\tikzset{edge/.style = {->,densely dotted,> = latex'}}
\draw[edge] (7) to node[near start, above left] {$c_2$} (2);
\draw[edge] (5)[in=300,out=30,looseness=1.5] to node[right] {$c_3$} (2);
\draw[edge] (4)[in=320,out=220,looseness=2] to node[below] {$c_2$} (2);
\end{tikzpicture}
\caption{\label{aetassG} Adding edges to the associated weighted digraph from Figure \ref{assG}(b), edges $(w,1)$ are dashed, and edges $(w,2)$ are dotted. To complete the graph, edges $(w,i)$ with $3\leq i\leq 8$ must be added.}
\end{center}
\end{figure}
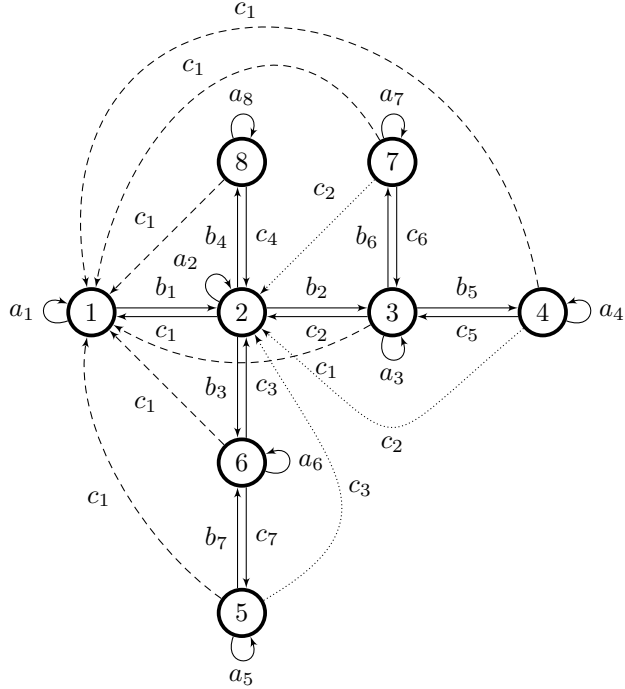

\begin{example} \label{exaTr2}
The adjacency matrix of the associated weighted complete digraph $\overline{G}$ to the  tree in Example \ref{exaTr} is
\begin{equation} \label{MA}
\A=\begin{pmatrix}
a_{{1}}&b_{{1}}&b_{{2}}&b_{{5}}&c_{{7}}&b_{{3}}&b_{{6}}&b_{{4}}\\
c_{{1}}&a_{{2}}&b_{{2}}&b_{{5}}&c_{{7}}&b_{{3}}&b_{{6}}&b_{{4}}\\
c_{{1}}&c_{{2}}&a_{{3}}&b_{{5}}&c_{{7}}&b_{{3}}&b_{{6}}&b_{{4}}\\
c_{{1}}&c_{{2}}&c_{{5}}&a_{{4}}&c_{{7}}&b_{{3}}&b_{{6}}&b_{{4}}\\
c_{{1}}&c_{{3}}&b_{{2}}&b_{{5}}&a_{{5}}&b_{{7}}&b_{{6}}&b_{{4}}\\
c_{{1}}&c_{{3}}&b_{{2}}&b_{{5}}&c_{{7}}&a_{{6}}&b_{{6}}&b_{{4}}\\
c_{{1}}&c_{{2}}&c_{{6}}&b_{{5}}&c_{{7}}&b_{{3}}&a_{{7}}&b_{{4}}\\
c_{{1}}&c_{{4}}&b_{{2}}&b_{{5}}&c_{{7}}&b_{{3}}&b_{{6}}&a_{{8}}
\end{pmatrix}.
\end{equation}
\end{example}

\begin{definition}[{\bf Tree-system}]
Let $T$ be a tree on $n$ vertices, and let $\A$ be the adjacency matrix of the associated weighted complete digraph of $T$. We define an $n$-dimensional homogeneous Lotka-Volterra system by
\begin{equation} \label{Tsy}
\frac{{\rm d} x_i}{{\rm d} t} =  x_i \sum_j A_{i,j} x_j,\qquad i=1,\ldots,n,
\end{equation}
to which we will refer as a tree-system, or a $T$-system if it is clear that $T$ is a tree.\footnote{Tree-systems are not to be confused with T-systems where T stands for Transfer matrix (or Toda or Tau).}
\end{definition}

\begin{example} \label{exaTr3}
The tree $T$ in Example \ref{exaTr} gives rise to the $T$-system
\begin{equation} \label{tsys}
\begin{split}
\dot{x_{{1}}}&=x_{{1}} \left( a_{{1}}x_{{1}}+b_{{1}}x_{{2}}+b_{{2}}x_{{3}}+b_{{5}}x_{{4}}
+c_{{7}}x_{{5}}+b_{{3}}x_{{6}}+b_{{6}}x_{{7}}+b_{{4}}x_{{8}}
 \right)\\
\dot{x_{{2}}}&=x_{{2}}\left(
c_{{1}}x_{{1}}+a_{{2}}x_{{2}}+b_{{2}}x_{{3}}+b_{{5}}x_{{4}}
+c_{{7}}x_{{5}}+b_{{3}}x_{{6}}+b_{{6}}x_{{7}}+b_{{4}}x_{{8}}
 \right)\\
\dot{x_{{3}}}&=x_{{3}} \left(
c_{{1}}x_{{1}}+c_{{2}}x_{{2}}+a_{{3}}x_{{3}}+b_{{5}}x_{{4}}
+c_{{7}}x_{{5}}+b_{{3}}x_{{6}}+b_{{6}}x_{{7}}+b_{{4}}x_{{8}}
\right)\\
\dot{x_{{4}}}&=x_{{4}} \left(c_{{1}}x_{{1}}+c_{{2}}x_{{2}}+c_{{5}}x_{{3}}+a_{{4}}x_{{4}}
+c_{{7}}x_{{5}}+b_{{3}}x_{{6}}+b_{{6}}x_{{7}}+b_{{4}}x_{{8}}
\right)\\
\dot{x_{{5}}}&=x_{{5}} \left(
c_{{1}}x_{{1}}+c_{{3}}x_{{2}}+b_{{2}}x_{{3}}+b_{{5}}x_{{4}}
+a_{{5}}x_{{5}}+b_{{7}}x_{{6}}+b_{{6}}x_{{7}}+b_{{4}}x_{{8}} \right)\\
\dot{x_{{6}}}&=x_{{6}} \left(
c_{{1}}x_{{1}}+c_{{3}}x_{{2}}+b_{{2}}x_{{3}}+b_{{5}}x_{{4}}
+c_{{7}}x_{{5}}+ a_{{6}}x_{{6}}+b_{{6}}x_{{7}}+b_{{4}}x_{{8}} \right)\\
\dot{x_{{7}}}&=x_{{7}} \left(
c_{{1}}x_{{1}}+c_{{2}}x_{{2}}+c_{{6}}x_{{3}}+b_{{5}}x_{{4}}
+c_{{7}}x_{{5}}+b_{{3}}x_{{6}}+ a_{{7}}x_{{7}}+b_{{4}}x_{{8}} \right)\\
\dot{x_{{8}}}&=x_{{8}} \left(
c_{{1}}x_{{1}}+c_{{4}}x_{{2}}+b_{{2}}x_{{3}}+b_{{5}}x_{{4}}
+c_{{7}}x_{{5}}+b_{{3}}x_{{6}}+b_{{6}}x_{{7}}+ a_{{8}}x_{{8}}
\right).
\end{split}
\end{equation}
\end{example}

\section{Properties of associated weighted complete digraphs}
The associated weighted complete digraph of a tree has the following nice property. \begin{proposition}
Let $T$ be a tree, and let $\bar{G}$ be its associated weighted digraph. Let $(u,v)$ be an edge in $T$. For all $w\not\in\{u,v\}$, the edges $(u,w),(v,w)\in\bar{G}$ carry the same weight.
\end{proposition}

\begin{proof}
This is clear from the construction. The shortest path from $u$ to $w$ differs from the shortest path from $v$ to $w$ only by one edge, i.e. $(u,v)$ or $(v,u)$. The weights only depend on the last edges of the paths, which is the same in both cases.
\end{proof}

\begin{corollary} \label{cor1}
Let $T$ be a tree. The adjacency matrix $\A$ of its associated weighted complete digraph $\overline{G}$ is an $n\times n$ matrix, with elements $A_{i,i}=a_i$, for each edge $e_i=(u_i,v_i)\in T$ we have $A_{u_i,v_i}=b_i$ and $A_{v_i,u_i}=c_i$, and for all $w\not\in \{u_i,v_i\}$ we have $A_{u_i,w}=A_{v_i,w}$.
\end{corollary}

Lemma 1 of \cite{QTMK} states that the expression $P_{u,v} = \alpha x_u + \beta x_v$, with $\alpha\beta\neq0$, is a DP of LV-system (\ref{LVsys}) if and only if, for some constant $b$ and all $w\not\in \{u,v \}$, we have $A_{u,w}=A_{v,w}$, $b_u = b_v = b$, $\alpha (A_{v,v}-A_{u,v}) = \beta (A_{v,u}-A_{u,u})$ and $(A_{v,v}-A_{u,v})(A_{v,u}-A_{u,u})\neq 0$. Hence, the property of $\A$ stated in Corollary \ref{cor1} guarantees that the tree-system (\ref{Tsy}) admits additional Darboux polynomials. And this implies, using the rather general method of \cite[Section 2]{QTMK}, the tree-system admits integrals, cf. equations \eqref{ei1} and \eqref{ei2}.

As the existence of a solution does not guarantee uniqueness, one wonders whether there are other matrices $\A$ with the same property. We show this is not the case.

\begin{proposition}
In Corollary \ref{cor1}, the matrix $A$ with elements $A_{i,i}=a_i$, $A_{u_i,v_i}=b_i$ and $A_{v_i,u_i}=c_i$ for each edge $e_i=(u_i,v_i)\in T$, is uniquely defined by the condition that for all $w\not\in \{u_i,v_i\}$ we have $A_{u_i,w}=A_{v_i,w}$.
\end{proposition}

\begin{proof}
In $\A$, there are $N=n^2-n-2(n-1)$ unknowns, which satisfy a system of $(n-1)(n-2)=N$ linear equations. We claim that each equation is used to solve for one unknown. Suppose $A_{u,v}$ is unknown. Let $u,w_1,\ldots,w_r,v$ be the shortest path in $T$ from $u$ to $v$. The weight of $(w_r,v)$ is known. It follows from Definition \ref{dacdg} that this is the value of $A_{u,v}$. But the aim here is to show how it follows from the condition. The condition tells us that $A_{w_{r-1},v}=A_{w_r,v}$, $A_{w_{r-2},v}=A_{w_{r-1},v}$, $\ldots$, and, finally, that  $A_{u,v}=A_{w_1,v}$ which hence equals the weight of $(w_r,v)$. To find $A_{u,v}$, we have used one equation to solve for one unknown $r$ times.
\end{proof}

\begin{definition}
Let $\A^{I;J}$ denote the matrix $\A$ with rows $i\in I$ and columns $j\in J$ deleted. Its determinant $|\A^{I;J}|$ is called a minor.
\end{definition}

By deleting a vertex, a tree becomes a forest with several connected components (trees).

\begin{proposition}[Factorisation of minors of $\A$] \label{FOMA}
Let $S$ denote the set of vertices of the tree which contains vertex $j$ in the forest obtained by deleting vertex $i$. Let $m$ be the number of vertices $k\in S$ such that $\min(i,j)<k<\max(i,j)$. Then
\begin{equation} \label{FOM}
|\A^{i;j}|=(-1)^m|\A^{S\setminus\{j\}\cup\{i\};S}|\prod_{v\in S\setminus\{j\}}(a_v-A_{j,v}).
\end{equation}
\end{proposition}
\begin{proof}
Denote the $i$th row of $\A$ by $\A_i$. If the shortest path in $T$ from vertex $k$ to vertex $j$ passes through $m$ vertices $k,l,\ldots, j$, then $\A_{k}-\A_{j}$ has $m$ non-zero elements, i.e. in columns $k,l,\ldots, j$.
We write $S=S_1+S_2+\cdots$, where $S_m$ contains all vertices whose shortest path to $j$ contains $m+1$ vertices. Suppose $k\in S_1$ satisfies $k<\min(i,j)$. Then $\A^{i;j}_k-\A^{i;j}_j$ has only 1 non-zero entry, i.e. $a_k-A_{k,j}$ in column $k$ (the $j$-th column,
which contained the other non-zero entry, is deleted). Expanding the determinant in the $k$-th row, after row-operation $\A^{i;j}_k\mapsto \A^{i;j}_k-\A^{i;j}_j$, shows that
\begin{equation} \label{aij}
|\A^{i;j}|=\epsilon^k_{i,j}(a_k-A_{j,k})|\A^{\{k,i\};\{k,j\}}|,
\end{equation}
with $\epsilon^k_{i,j}=1$. When $k\in S_1$ satisfies $\max(i,j)<k$, after row-reduction of the $k-1$ row, the same non-zero entry arises in column $k-1$, and we find the same formula. If $\min(i,j)<k<\max(i,j)$ then we pick up a minus sign, $\epsilon^k_{i,j}=-1$. Continuing the row-reduction in this way, first for all $k\in S_1$, then $k\in S_2$, etc., yields formula (\ref{FOM}).
\end{proof}

\begin{proposition}[Ratios of minors of $\A$] \label{ROM}
For all edges $e_i=(u_i,v_i)$ and $w\not\in\{u_i,v_i\}$ we have
\[
(a_{u_i}-c_i)|\A^{w;u_i}|=(-1)^{u_i+v_i}(a_{v_i}-b_i)|\A^{w;v_i}|.
\]
\end{proposition}

\begin{proof}
The proposition follows as a consequence of Proposition \ref{FLCM}.
\end{proof}

\section{Integrals for tree-systems}
\begin{definition}
Define an $(n-1)\times n$ matrix $\B$, whose rows and columns are indexed by edges and vertices respectively, as follows: for each edge $e_i=(u_i,v_i)\in T$, $x\in\{u_i,v_i\}$, $w\not\in \{u_i,v_i\}$, let $B_{i,x}=a_x$, and $B_{i,w}=A_{x,w}$ (whose value does not depend on $x$, due to Corollary \ref{cor1}).
\end{definition}

With this definition the elements of the $i$-th row of the matrix $\B$ are the coefficients of the cofactor of the $i$-th Darboux polynomial,
\begin{equation} \label{Pi}
P_i=(c_i-a_{u_i})x_{u_i}+(a_{v_i}-b_i)x_{v_i},
\end{equation}
of the tree-system (\ref{Tsy}). Indeed, we have $\dot{P}_i=P_i\left(a_{u_i}x_{u_i} + a_{v_i}x_{v_i} + \sum_{w\neq u_i,v_i} A_{u_i,w}x_w\right)$, cf. the $\Leftarrow$ part of the proof of \cite[Lemma 1]{QTMK}, i.e. the cofactor of $P_i$ equals $\sum_{j=1}^n B_{i,j}x_j$.

Here is a slightly more general method than the one provided in \cite[Section 2]{QTMK}. Define $\v=\ln(\x)$ and $\w=\ln(\P)$. Then $\dot{\v}=\A \x$ and $\dot{\w}=\B \x$. Multiplying yields\footnote{Here $\adj{\A}$ is the adjugate matrix of $\A$, which is the transpose of its cofactor matrix. It satisfies $\A.\adj{\A}=|\A|\I$. Note that the word cofactor has acquired two meanings: while $\A$ is the matrix of coefficients of cofactors of the coordinate DPs, its cofactor matrix is its matrix of signed minors.} $\adj{\A}\dot{\v}=|\A|\x$ and $|\A|\dot{\w}=|\A|\B\x=\B|\A|\x=\B.\adj{\A}\dot{\v}$. Integrating the latter yields $n-1$ integrals $\mathbf{I}=|\A|\ln(\P)-\B.\adj{\A}\ln(\x)$, and by exponentiation the $i$-th one becomes
\begin{equation} \label{ei1}
P_i^{|\A|}\prod_{k=1}^n x_k^{Z_{i,k}},
\end{equation}
where
\begin{equation} \label{ei2}
\Z=-\B.\adj{\A}.
\end{equation}

Note that we assume that the $3n-2$ parameters in $\A$ have generic values. This ensures e.g. that 
\[
|A|\prod_{i=1}^{n-1}\left(b_i-a_{v_i}\right)\left(c_i-a_{u_i}\right)\neq 0,
\]
so that neither the exponent of $P_i$ in (\ref{ei1}) nor the coefficients of $x_{u_i}$ and $x_{v_i}$ in (\ref{Pi}) vanish.
In appendix A, we show that in the particular case where $|\A|=0$, the expression \eqref{ei1} provides only one functionally independent integral.
  
\begin{example} \label{exaTr4}
For our running example, the matrix $\B$ is given by
\begin{equation} \label{MB}
\B=
\begin{pmatrix}
a_{{1}}&a_{{2}}&b_{{2}}&b_{{5}}&c_{{7}}&b_{{3}}&b_{{6}}&b_{{4}}\\
c_{{1}}&a_{{2}}&a_{{3}}&b_{{5}}&c_{{7}}&b_{{3}}&b_{{6}}&b_{{4}}\\
c_{{1}}&a_{{2}}&b_{{2}}&b_{{5}}&c_{{7}}&a_{{6}}&b_{{6}}&b_{{4}}\\
c_{{1}}&a_{{2}}&b_{{2}}&b_{{5}}&c_{{7}}&b_{{3}}&b_{{6}}&a_{{8}}\\
c_{{1}}&c_{{2}}&a_{{3}}&a_{{4}}&c_{{7}}&b_{{3}}&b_{{6}}&b_{{4}}\\
c_{{1}}&c_{{2}}&a_{{3}}&b_{{5}}&c_{{7}}&b_{{3}}&a_{{7}}&b_{{4}}\\
c_{{1}}&c_{{3}}&b_{{2}}&b_{{5}}&a_{{5}}&a_{{6}}&b_{{6}}&b_{{4}}
\end{pmatrix}.
\end{equation}
The integrals can be calculated explicitly but, with 22 free parameters, each of them would take up too much space to be included explicitly here.
\end{example}

\begin{example}[Special case]
Considering the special case where
\[
a_i=i,\quad b_i=i+5,\quad c_i=i+11
\]
we find the 7 Darboux polynomials
\[
11 x_{1}-4 x_{2},\ 11 x_{2}-4 x_{3},\ 12 x_{2}-2 x_{6},\ 13 x_{2}-x_{8},\ 13 x_{3}-6 x_{4},\ 14 x_{3}-4 x_{7},\ 13 x_{5}-6 x_{6}.
\]
The cofactor of the first DP is $(\B\x)_1=x_1+2x_2+7x_3+10x_4+18x_5+8x_6+11x_7+9x_8$ and the integral corresponding to the first edge/Darboux polynomial (after inversion and taking an 8-th root) is
\[
\frac{\left(11 x_{1}-4 x_{2}\right)^{1393111} x_{3}^{169884} x_{4}^{22880} x_{5}^{19008} x_{6}^{92928} x_{7}^{37752} x_{8}^{123552}}{x_{1}^{1378135} x_{2}^{467252}}.
\]
\end{example}

\begin{proposition}[Factorisation of linear combinations of minors of $\A$] \label{FLCM}
Let $\Z=-\B.\adj{\A}$. For all edges $e_i=(u_i,v_i)\in T$ we have
\begin{equation}\label{Zij}
Z_{i,k} =
\begin{cases}
(-1)^{k+u_i}(c_i-a_{u_i})|\A^{k;u_i}|, & k\neq v_i \\
(-1)^{k+v_i}(b_i-a_{v_i})|\A^{k;v_i}|, & k\neq u_i.
\end{cases}
\end{equation}
\end{proposition}
Note that when $k\not\in\{u_i,v_i\}$ the two expressions in (\ref{Zij}) are equal, which is the statement of Proposition \ref{ROM}.
\begin{proof}
Let the cofactor matrix of $\A$ be $\C$, i.e. $C_{i,j}=(-1)^{i+j}|\A^{i;j}|$. Then we have $\adj{\A}=\C^t$, or, for all $i,k$,
\begin{equation} \label{AC}
\sum_j A_{i,j}C_{k,j}=\begin{cases}
|\A|, & i=k \\
0, & i\neq k.
\end{cases}
\end{equation}
We need to show that
\begin{equation} \label{BC}
-Z_{i,k}=\sum_j B_{i,j}C_{k,j}=\begin{cases}
(a_{u_i} - c_i)C_{k,u_i}, & k\neq v_i \\
(a_{v_i} - b_i)C_{k,v_i}, & k\neq u_i.
\end{cases}
\end{equation}

We notice that $B_{i,j}-A_{u_i,j}=(a_{v_i}-b_i)\delta_{j,v_i}$, and
$B_{i,j}-A_{v_i,j}=(a_{u_i}-c_i)\delta_{j,u_i}$, in terms of Kronecker's delta, $\delta_{i,i}=1$ and $\delta_{i,j\neq i}=0$. When $k\neq v_i$ we find
\[
-Z_{i,k}=\sum_j B_{i,j}C_{k,j}-\sum_j A_{v_i,j}C_{k,j}=(a_{u_i}-c_i)C_{k,u_i}
\]
and, when $k\neq u_i$,
\[
-Z_{i,k}=\sum_j B_{i,j}C_{k,j}-\sum_j A_{u_i,j}C_{k,j}=(a_{v_i}-b_i)C_{k,v_i}.
\]
\end{proof}

Using the above result, it now follows that if a graph contains two adjacent edges, the corresponding LV system possesses an integral that only depends on the three vertices involved, and not on any of the other vertices of the graph. Consider two edges $e_i=(u_i,v_i)$ and $e_j=(u_j,v_j)$. Then $e_i$ and $e_j$ are adjacent if either $a)\ v_i=u_j$, $b)\ v_i=v_j$, $c)\ u_i=u_j$, or $d)\ u_i=v_j$.  

\begin{corollary}[Three-point integrals] \label{TPI}
Let $\Z=-\B.\adj{\A}$.
\begin{enumerate}[a)]
\item Suppose $v_i=u_j$. The edges $e_i$ and $e_j$ give rise to integrals
\[
I_i=P_i^{|\A|}\prod_{k=1}^n x_k^{Z_{i,k}},\quad
I_j=P_j^{|\A|}\prod_{k=1}^n x_k^{Z_{j,k}},
\]
where
\[
P_i=(c_i-a_{u_i})x_{u_i}+(a_{u_j}-b_i)x_{u_j},\quad
P_j=(c_j-a_{u_j})x_{u_j}+(a_{v_j}-b_j)x_{v_j}.
\]
The ratio
\begin{equation}\label{rat1}
\sqrt[|\A|]{I_i^{a_{u_j}-c_j}/I_j^{a_{u_j}-b_i}}
=\left(P_i/x_{u_i}\right)^{a_{u_j}-c_j}\left(P_j/x_{v_j}\right)^{b_i-a_{u_j}}
\end{equation}
depends on $x_{u_i}/x_{u_j},x_{u_j}/x_{v_j}$ only.
\item When $v_i=v_j$ then
\begin{equation}\label{rat2}
\sqrt[|\A|]{{I_i^{a_{v_j}-c_j}}/{I_j^{a_{v_j}-b_i}}}
=\left(P_i/x_{u_i}\right)^{a_{v_j}-b_j}\left(P_j/x_{u_j}\right)^{b_i-a_{v_j}}
\end{equation}
depends on $x_{u_i}/x_{v_j},x_{v_j}/x_{u_j}$ only.
\item And, when $u_i=u_j$ then
\begin{equation}\label{rat3}
\sqrt[|\A|]{{I_i^{a_{u_j}-c_j}}/{I_j^{a_{u_j}-b_i}}}
=\left(P_i/x_{v_i}\right)^{a_{u_j}-c_j}\left(P_j/x_{v_j}\right)^{c_i-a_{u_j}}
\end{equation}
depends on $x_{v_i}/x_{u_j},x_{u_j}/x_{v_j}$ only.
\end{enumerate}
\end{corollary}
\begin{proof}
We provide a proof for (\ref{rat1}), the cases (\ref{rat2}) and (\ref{rat3}) are similar. Note that when $u_i=v_j$ the case (\ref{rat1}) applies after interchanging $i$ and $j$. Also note that any function of integrals is an integral. 

We simplify the ratio
\begin{align}
\frac{I_i^{a_{u_j}-c_j}}{I_j^{a_{u_j}-b_i}}
&=\frac{P_i^{|\A|(a_{u_j}-c_j)}}{P_j^{|\A|(a_{u_j}-b_i)}}
\prod_{k=1}^n x_k^{(a_{u_j}-c_j)Z_{i,k}-(a_{u_j}-b_i)Z_{j,k}}\notag\\
&=\frac{P_i^{|\A|(a_{u_j}-c_j)}}{P_j^{|\A|(a_{u_j}-b_i)}}
\prod_{k\in\{u_i,v_j\}} x_k^{(a_{u_j}-c_j)Z_{i,k}-(a_{u_j}-b_i)Z_{j,k}}\label{IoI}
\end{align}
as we have, due to Proposition \ref{FLCM}, for $k\neq \{u_i,v_j\}$,
\begin{align*}
(a_{u_j}-c_j)Z_{i,k}-(a_{u_j}-b_i)Z_{j,k}&=
(a_{u_j}-c_j)(-1)^{k+u_j}(b_i-a_{u_j})|\A^{k;u_j}|\\
& \ \ \ -(a_{u_j}-b_i)(-1)^{k+u_j}(c_j-a_{u_j})|\A^{k;u_j}|\\
&=0.
\end{align*}
The exponent of $x_{u_i}$ in (\ref{IoI}) can be written as
\[
(a_{u_j}-c_j)Z_{i,u_i}-(a_{u_j}-b_i)Z_{j,u_i}=(c_j-a_{u_j})|\A|.
\]
This can be seen by subtracting the $u_j$-th row of $\A$ from the $u_i$-th row, expanding the determinant in the $u_i$-th row, i.e.
\[
|\A|=(a_{u_i}-c_i)|\A^{u_i;u_i}|+(-1)^{u_i+u_j}(b_i-a_{u_j})|\A^{u_i;u_j}|,
\]
and using Proposition \ref{FLCM}. Similarly, it can be shown that the exponent of $x_{v_j}$ in (\ref{IoI}) can be written as
\[
(a_{u_j}-c_j)Z_{i,v_j}-(a_{u_j}-b_i)Z_{j,v_j}=(a_{u_j}-b_i)|\A|.
\]
The result follows by combining terms and taking the $|\A|$-th root.
\end{proof}

\begin{example} \label{eae}
The T-system \eqref{tsys} has the following integrals (which do not take up too much space), related to the pairs of adjacent edges $(e_1,e_4),\
(e_3,e_4),\
(e_3,e_7),\
(e_2,e_3),\
(e_2,e_5),\
(e_2,e_6)$:
\begin{align*}
&\left(\left(c_{1}-a_{1}\right)+\left(a_{2}-b_{1}\right) \dfrac{x_{2}}{x_{1}}\right)^{a_{2}-c_{4}}\left(\left(c_{4}-a_{2}\right) \dfrac{x_{2}}{x_{8}}+\left(a_{8}-b_{4}\right)\right)^{b_1-a_{2}},\\
&\left(\left(c_{3}-a_{2}\right)\dfrac{x_{2}}{x_{6}}+\left(a_{6}-b_{3}\right) \right)^{a_{2}-c_{4}}\left(\left(c_{4}-a_{2}\right) \dfrac{x_{2}}{x_{8}}+\left(a_{8}-b_{4}\right)\right)^{c_3-a_{2}},\\
&\left(\left(c_{3}-a_{2}\right) +\left(a_{6}-b_{3}\right) \dfrac{x_{6}}{x_{2}}\right)^{a_{6}-b_{7}}\left(\left(c_{7}-a_{5}\right) +\left(a_{6}-b_{7}\right)\dfrac{ x_{6}}{x_{5}}\right)^{b_3-a_{6}},
\end{align*}
\begin{equation}\label{ints}
\begin{split}
&\left(\left(c_{2}-a_{2}\right)\dfrac{x_{2}}{x_{3}}+\left(a_{3}-b_{2}\right) \right)^{a_{2}-c_{3}}\left(\left(c_{3}-a_{2}\right) \dfrac{x_{2}}{x_{6}}+\left(a_{6}-b_{3}\right) \right)^{c_2-a_{2}},\\
&\left(\left(c_{2}-a_{2}\right)+\left(a_{3}-b_{2}\right) \dfrac{x_{3}}{x_{2}}\right)^{a_{3}-c_{5}}\left(\left(c_{5}-a_{3}\right) \dfrac{x_{3}}{x_{4}}+\left(a_{4}-b_{5}\right) \right)^{b_2-a_{3}},\\
&\left(\left(c_{2}-a_{2}\right) +\left(a_{3}-b_{2}\right) \dfrac{x_{3}}{x_{2}}\right)^{a_{3}-c_{6}}\left(\left(c_{6}-a_{3}\right) \dfrac{x_{3}}{x_{7}}+\left(a_{7}-b_{6}\right) \right)^{b_2-a_{3}}.
\end{split}
\end{equation}
\end{example}

The adjacent edges in Example \ref{eae} were chosen in such a way that the resulting integrals are functionally independent. This can be done algorithmically, which will be explained in the proof of Theorem \ref{tfii}. We reiterate that 3-point integrals are local functions, in the sense that they depend on only a subset of the variables. To prove that they are invariant, we only need those components of the system of differential equations which describe the evolution of the subset of variables. Local integrals were employed before in \cite{KMQ}.

\section{Solution by quadratures using edge-variables}
Given a tree $T$ on $n$ vertices, we introduce for each edge $e_i=(u_i,v_i)$ a so-called edge variable $y_i=x_{u_i}/x_{v_i}$. With edge variables $y_1,\ldots, y_{n-1}$, and a distinguished vertex variable $x_k$, we can express the vertex variables as follows.
\begin{equation} \label{trans}
x_i=x_k \prod_{j\in {\cal E}_i^k} y_j^{\epsilon_j},
\end{equation}
where $j\in{\cal E}_i^k$ iff either $e_j=(u_j,v_j)$ (with $u_j<v_j$) or $e_j^\prime=(v_j,u_j)$ is an edge in the directed tall tree which is the shortest path from $x_i$ to $x_k$, and if it is $e_j$ then $\epsilon_j=1$, if it is $e_j^\prime$ then $\epsilon_j=-1$.
\begin{proposition} \label{tev}
A T-system can be written in edge variables, and the distinguished $x_k$, as
\[
\dot{y_i}=x_k \left(\prod_{j\in {\cal E}_{u_i}^k} y_j^{\epsilon_j} \right)\left(\left(a_{u_i}-c_i\right)y_i-a_{v_i}+b_i \right),\qquad i=1,\ldots,n-1
\]
and
\[
\dot{x_k}=x_k^2\sum_{i=1}^n d_i \prod_{j\in {\cal E}_i^k} y_j^{\epsilon_j}
\]
with $d_k=a_k$ and, for $i\neq k$,
\[
d_i=\begin{cases}
b_j& \text{ if } e_j=(w,x_i)\\
c_j& \text{ if } e_j=(x_i,w),
\end{cases}
\]
where $x_i,w,\ldots,x_k$ is the shortest path from $x_i$ to $x_k$.
\end{proposition}
\begin{proof}
Firstly, we have
$\begin{aligned}[t]
\dot{y_i}&=\frac{\dot{x_{u_i}}}{x_{v_i}}-\frac{x_{u_i}}{x_{v_i}^2}\dot{x_{v_i}}\\
&=y_i\left(a_{{u_i}}x_{u_i}+b_ix_{v_i}-c_ix_{u_i}-a_{{v_i}}x_{v_i}\right)\\
&=x_{u_i}\left(\left(a_{{u_i}}-c_i\right)y_{i}-a_{{v_i}}+b_i\right),\\
&=x_k \left(\prod_{j\in {\cal E}_{u_i}^k} y_j^{\epsilon_j}\right)\left(\left(a_{{u_i}}-c_i\right)y_{i}-a_{{v_i}}+b_i\right).
\end{aligned}$

Next we consider the equation for $x_k$. From Definition \ref{dacdg}
\[
\dot{x_k}=x_k \sum_{i=1}^n d_i x_i,
\]
and the result follows from \eqref{trans}.
\end{proof}

Interestingly, the second line of the proof of Proposition \ref{tev} can be written as
$\dot{y_i}=-y_iP_i$, suggesting that the DPs of system (\ref{Tsy}) play the role of cofactors of the edge variables. However, one has to be careful here. Expressed in edge variables, the rational functions $P_i$ may have $y_i$ as a factor of its denominator. Hence, the resulting edge system does not necessarily take a Lotka-Volterra form.

\begin{example}
Choosing $x_8$ as the distinguished variable, the vertex variables relate to the edge variables by
\[
\frac{x_{1}}{x_{8}} = y_{1} y_{4},\
\frac{x_{2}}{x_{8}} = y_{4},\
\frac{x_{3}}{x_{8}} = \frac{y_{4}}{y_{2}},\
\frac{x_{4}}{x_{8}} = \frac{y_{4}}{y_{2} y_{5}},\
\frac{x_{5}}{x_{8}} = \frac{y_{4} y_{7}}{y_{3}},\
\frac{x_{6}}{x_{8}} =\frac{y_{4}}{y_{3}},\
\frac{x_{7}}{x_{8}} = \frac{y_{4}}{y_{2} y_{6}}
\]
Using edge variables, the system \eqref{tsys} transforms into
\begin{equation}\label{rsys}
\begin{pmatrix}
\dot{y_1}\\
\dot{y_2}\\
\dot{y_3}\\
\dot{y_4}\\
\dot{y_5}\\
\dot{y_6}\\
\dot{y_7}
\end{pmatrix}=y_{4}x_{8}
\begin{pmatrix}
y_{1} \left(a_{1} y_{1}-c_{1} y_{1}-a_{2}+b_{1}\right)\\
a_{2} y_{2}-c_{2} y_{2}-a_{3}+b_{2}\\
a_{2} y_{3}-c_{3} y_{3}-a_{6}+b_{3}\\
a_{2} y_{4}-c_{4} y_{4}-a_{8}+b_{4}\\
\left(a_{3}y_{5}-c_{5}y_{5}-a_{4}+b_{5}\right)/y_{2}\\
\left(a_{3} y_{6}-c_{6} y_{6}-a_{7}+b_{6}\right)/y_{2}\\
y_{7}\left(a_{5} y_{7}-c_{7} y_{7}-a_{6}+b_{7}\right)/y_{3}
\end{pmatrix},
\end{equation}
together with
\begin{equation}\label{x8}
\dot{x_8}=y_4 x_{8}^2 \left(c_{1} y_{1}+c_{4}+\frac{c_{7} y_{7}}{y_{3}}+\frac{a_{8}}{y_4}+\frac{b_{2}}{y_{2}}+\frac{b_{3} }{y_{3}}+\frac{b_{5} }{y_{2} y_{5}}+\frac{b_{6} }{y_{2} y_{6}}\right).
\end{equation}
The integrals \eqref{ints} are expressed in edge variables as
\begin{equation}\label{intse}
\begin{split}
&\left(\left(c_{1}-a_{1}\right)+\left(a_{2}-b_{1}\right) /y_1\right)^{a_{2}-c_{4}}\left(\left(c_{4}-a_{2}\right) y_4+\left(a_{8}-b_{4}\right)\right)^{b_1-a_{2}},\\
&\left(\left(c_{3}-a_{2}\right)y_3+\left(a_{6}-b_{3}\right) \right)^{a_{2}-c_{4}}\left(\left(c_{4}-a_{2}\right) y_4+\left(a_{8}-b_{4}\right)\right)^{c_3-a_{2}},\\
&\left(\left(c_{3}-a_{2}\right) +\left(a_{6}-b_{3}\right) /y_3\right)^{a_{6}-b_{7}}\left(\left(c_{7}-a_{5}\right) +\left(a_{6}-b_{7}\right)/y_7\right)^{b_3-a_{6}},\\
&\left(\left(c_{2}-a_{2}\right)y_2+\left(a_{3}-b_{2}\right) \right)^{a_{2}-c_{3}}\left(\left(c_{3}-a_{2}\right) y_3+\left(a_{6}-b_{3}\right) \right)^{c_2-a_{2}},\\
&\left(\left(c_{2}-a_{2}\right)+\left(a_{3}-b_{2}\right) /y_2\right)^{a_{3}-c_{5}}\left(\left(c_{5}-a_{3}\right) y_5+\left(a_{4}-b_{5}\right) \right)^{b_2-a_{3}},\\
&\left(\left(c_{2}-a_{2}\right) +\left(a_{3}-b_{2}\right) /y_2\right)^{a_{3}-c_{6}}\left(\left(c_{6}-a_{3}\right) y_6+\left(a_{7}-b_{6}\right) \right)^{b_2-a_{3}}.
\end{split}
\end{equation}
\end{example}

\begin{theorem} \label{nm1}
Each $T$-system, where $T$ is a tree on $n$ vertices, reduces to an $(n-1)$-dimensional edge system which is superintegrable, and solvable by quadratures.
\end{theorem}

\begin{example}
Denote the r.h.s. of the combined system (\ref{rsys},\ref{x8}) by $f$, so that
\begin{equation} \label{ex8}
(\dot{y},\dot{x_8})=f=ag,
\end{equation}
with $a=y_{4}x_{8}$. Let us denote the 6 integrals \eqref{intse} of the tree-system by $K_i$. They are integrals of (\ref{ex8}), i.e. we have $f\cdot \nabla K_i=0$.
But then, since $f=ag$, we have $g\cdot \nabla K_i=0$ and hence the integrals are also integrals of the system
\begin{equation} \label{ex82}
(\dot{y},\dot{x_8})=g.
\end{equation}
Because the integrals depend on the $y$-variables only, and the first seven components of $g$ do not depend on $x_8$, we find that the 7-dimensional system
\begin{equation}\label{esys}
\begin{pmatrix}
\dot{y_1}\\
\dot{y_2}\\
\dot{y_3}\\
\dot{y_4}\\
\dot{y_5}\\
\dot{y_6}\\
\dot{y_7}
\end{pmatrix}=
\begin{pmatrix}
y_{1} \left((a_{1} -c_{1}) y_{1}-a_{2}+b_{1}\right)\\
(a_{2} -c_{2} )y_{2}-a_{3}+b_{2}\\
(a_{2} -c_{3} )y_{3}-a_{6}+b_{3}\\
(a_{2} -c_{4} )y_{4}-a_{8}+b_{4}\\
\left((a_{3}-c_{5})y_{5}-a_{4}+b_{5}\right)/y_{2}\\
\left((a_{3}-c_{6} )y_{6}-a_{7}+b_{6}\right)/y_{2}\\
y_{7}\left((a_{5}-c_{7} )y_{7}-a_{6}+b_{7}\right)/y_{3}
\end{pmatrix},
\end{equation}
admits 6 functionally independent integrals and hence is superintegrable.
\end{example}

\begin{proof}[Proof of Theorem \ref{nm1}.] We rewrite the tree-system in variables $(y,x_k)$, cf. Proposition \ref{tev}, as
\begin{equation} \label{yxk}
\begin{split}
\dot{y_i}&=x_kf_i, \quad i=1,\ldots,n-1\\
\dot{x_k}&=x_k^2g,
\end{split}
\end{equation}
where $f_i,g$ do not depend on $x_k$. The $n-2$ three-point integrals given in Corollary \ref{TPI} depend on the $y$-variables only and they are functionally independent integrals of the $(n-1)$-dimensional system
\begin{equation} \label{edgesys}
\dot{y_i}=f_i, \quad i=1,\ldots,n-1,
\end{equation}
which we call the edge system. We note that one may divide each $f_i$ by the greatest common divisor of all the $f_i$ (which is nontrivial if $x_k$ is at the end of a branch of the tree).

Each equation of the edge system \eqref{edgesys} is either of the form
\begin{equation} \label{form}
\dot{y}=p(\alpha y+\beta)\quad \text{ or }\quad \dot{y}=py(\alpha y+\beta),
\end{equation}
where $\alpha,\beta$ are constants, and $p$ is a function of $t$.
In the first case we get
\[
\frac{\dot{y}}{\alpha y+\beta}=p \Rightarrow \frac{\ln(\alpha y+\beta)}{\alpha} = \int p.
\]
In the second case we get, with $z=1/y$,
\[
p=\frac{\dot{y}}{y(\alpha y+\beta)}=-\frac{\dot{z}}{\alpha+\beta z}
\Rightarrow \frac{\ln(\alpha+\beta/y)}{\beta} = - \int p.
\]
Starting at the distinguished variable $x_k$, the equations for the edge variables connected to $x_k$ are of the form \eqref{form}. Once these equations have been solved, one considers the edges one step further away from $x_k$, and this process can be continued.
\end{proof}

\begin{example}
In the system (\ref{esys}), the first 4 equations are of the form \eqref{form}. We obtain
\begin{align*}
&y_1=\frac{a_{2}-b_{1}}{C_1{\mathrm e}^{\left(a_{2}-b_{1}\right) t} +a_{1}-c_{1}},
&y_2=C_2 {\mathrm e}^{\left(a_{2}-c_{2}\right) t} +\frac{a_{3}-b_{2}}{a_{2}-c_{2}},\\
&y_3=C_3 {\mathrm e}^{\left(a_{2}-c_{3}\right) t} +\frac{a_{6}-b_{3}}{a_{2}-c_{3}},
&y_4=C_4 {\mathrm e}^{\left(a_{2}-c_{4}\right) t} +\frac{a_{8}-b_{4}}{a_{2}-c_{4}}.
\end{align*}
Given the solutions for $y_2$ and $y_3$ we find
\begin{align*}
y_5&=C_5 \left( C_2+{\frac { \left( a_{{3}}-b_{{2}} \right) {
{\rm e}^{- \left( a_{{2}}-c_{{2}} \right) t}}}{a_{{2}}-c_{{2}}}}
 \right) ^{-{\frac {a_{{3}}-c_{{5}}}{a_{{3}}-b_{{2}}}}}-{\frac {b_{{5}
}-a_{{4}}}{a_{{3}}-c_{{5}}}}
\\
y_6&=C_6 \left( C_2+{\frac { \left( a_{{3}}-b_{{2}} \right) {
{\rm e}^{- \left( a_{{2}}-c_{{2}} \right) t}}}{a_{{2}}-c_{{2}}}}
 \right) ^{-{\frac {a_{{3}}-c_{{6}}}{a_{{3}}-b_{{2}}}}}-{\frac {b_{{6}
}-a_{{7}}}{a_{{3}}-c_{{6}}}}
\\
y_7&={(b_{{7}}-a_{{6}}) \left(  C_7 \left( C_3+{\frac { \left( a
_{{6}}-b_{{3}} \right) {{\rm e}^{- \left( a_{{2}}-c_{{3}} \right) t}}
}{a_{{2}}-c_{{3}}}} \right) ^{{\frac {b_{{7}}-a_{{6}}}{a_{{6}}-b_{{3}}
}}}-a_{{5}}+c_{{7}} \right) ^{-1}}.
\end{align*}
\end{example}

Once a solution to the $n-1$ dimensional edge system has been found, one may retrieve a solution to the original tree-system as follows. Let the tree-system, cf. \eqref{yxk}, be
\begin{equation} \label{tysys}
\begin{split}
\dot{y_i}&= h f_i, \quad i=1,\ldots,n-1\\
\dot{x_k}&=x_k h g,
\end{split}
\end{equation}
and suppose a solution to the edge system
\begin{equation} \label{eesys}
\dot{y_i}=f_i, \quad i=1,\ldots,n-1
\end{equation}
is known. One first determines the solution to $\dot{x_k}=x_k g(y)$, from $\ln(x_k)=\int g$. With $z=(y,x_k)$ and $w=(f,x_k g)$ we then have the solution,
$z(\tau)$ to
\begin{equation} \label{eg}
\frac{d z}{d \tau}=w(z).
\end{equation}
We seek to find a solution to
\begin{equation} \label{eag}
\frac{d z}{d t}=h(z)w(z).
\end{equation}
From \eqref{eg} and \eqref{eag} we obtain
\[
\frac{d t}{d \tau}=\frac{1}{h(z(\tau))} \implies t=\int \frac{1}{h(z(\tau))} d\tau := u(\tau).
\]
Inverting $u$ we find $\tau=u^{-1}(t)$ and the solution to \eqref{eag} is $z(u^{-1}(t))$.

\section{Tree-systems are superintegrable}
\begin{theorem} \label{tfii}
The $n$-dimensional tree-system \eqref{Tsy} admits $n-1$ functionally independent integrals.
\end{theorem}
\begin{proof}
Pick an edge of $T$, call it $\bar{e}_{n-1}$. Let $T_{n-1}\subset T$ be the subtree on 2 vertices with edge $\bar{e}_{n-1}$.
Choose an edge $\bar{e}_{n-2}\in T$ which is adjacent to $T_{n-1}$. Let $T_{n-1}\subset T$ be the subtree on 3 vertices with edges $\bar{e}_{n-1},\bar{e}_{n-2}$. Continue the process, i.e. choose edges $\bar{e}_{n-i}\in T$ adjacent to $T_{n-i+1}$, and let
$T_{n-i}\subset T$ be the tree on $i+1$ vertices with edges $\bar{e}_{n-1},\ldots,\bar{e}_{n-i}$, for $i=3,\ldots,n-1$, so that $T_1=T$.
Now consider the set $K=\{K_1,K_2,\ldots,K_{n-1}\}$, where $K_1$ is any integral of the form \eqref{ei1}, and $K_i$, $i>1$, is the 3-point integral which depends on edge-variable $\bar{y}_{i-1}$ and the edge-variable of the edge it is adjacent to in $T_i$. The integral $K_1$ can be written in terms of a distinguished variable $x_k$ and the edge-variables $\bar{y}_1,\ldots,\bar{y}_{n-1}$. By construction, the Jacobian matrix of $K$ with respect to these variables is `upper-triangular' and hence the $n-1$ integrals are functionally independent.
\end{proof}

\begin{example}
For the tree-system \eqref{tsys}, if we choose
\begin{align*}
\bar{e}_7=(3,7),\quad
&\bar{e}_6=(2,3),\quad
\bar{e}_5=(3,4),\quad
\bar{e}_4=(2,6),\\
\bar{e}_3=(5,6),\quad
&\bar{e}_2=(2,8),\quad
\bar{e}_1=(1,2),\quad
\end{align*}
the integrals $K_2,\ldots,K_7$ correspond to \eqref{ints} and the Jacobian matrix takes the form
\[
\begin{pmatrix}
\ast & \ast & \ast & \ast & \ast & \ast & \ast & \ast \\
0 & \ast & \ast & 0 & 0 & 0 & 0 & 0 \\
0 & 0 & \ast & 0 & \ast & 0 & 0 & 0 \\
0 & 0 & 0 & \ast & \ast & 0 & 0 & 0 \\
0 & 0 & 0 & 0 & \ast & 0 & \ast & 0 \\
0 & 0 & 0 & 0 & 0 & \ast & \ast & 0 \\
0 & 0 & 0 & 0 & 0 & 0 & \ast & \ast
\end{pmatrix}.
\]
\end{example}

\vspace{5mm}
\noindent
{\bf Acknowledgement}

We are grateful to our colleague Ian Marquette for pointing us to reference \cite{Maier}, and to two anonymous referees for several useful suggestions.

\appendix

\section{The rank of $\Z$ when $|\A|=0$.}
The following theorem implies that for special values of the parameters where $|\A|=0$ only 1 of the $n-1$ integrals is functionally independent. \begin{theorem}\label{zhro}
Let $\A$ be the $n\times n$ adjacency matrix of the associated weighted complete digraph to a tree on $n$ vertices, as in Definition \ref{dacdg}, and let $\B$ be the $(n-1)\times n$ matrix of coefficients of the cofactors of the additional Darboux polynomials \eqref{Pi} of the corresponding Lotka-Volterra tree-system. If $|\A|=0$, then the matrix $\Z=-\B.\adj{\A}$ has rank 1.
\end{theorem}
To prove Theorem \ref{zhro} we need to consider the missing cases in Proposition \ref{ROM}, i.e., where $w\in\{u_i,v_i\}$.
\begin{proposition} \label{FAE}
For all edges $e_i=(u_i,v_i)$ and $w\in\{u_i,v_i\}$ we have
\[
(a_{u_i}-c_i)|\A^{w;u_i}|-(-1)^{u_i+v_i}(a_{v_i}-b_i)|\A^{w;v_i}|=\begin{cases}
|\A| & \text{ if } w=u_i\\
-(-1)^{u_i+v_i}|\A| & \text{ if } w=v_i.
\end{cases}
\]
\end{proposition}

\begin{proof}
We first prove the case $w=u_i$. Let $\bar{\A}$ be the matrix obtained from $\A$ by replacing the $u_i$-th row by the $v_i$-th row.
By expanding in this row, we find, as $\A_{v_i,j}=\A_{u_i,j}$ when $j\not\in\{u_i,v_i\}$,
\begin{equation}\label{bAiu}
0=|\bar{\A}|=\left(\sum_{j\not\in\{u_i,v_i\}}(-1)^{u_i+j}A_{u_i,j}|\A^{u_i;j}|\right) +c_i|\A^{u_i;u_i}|+(-1)^{u_i+v_i}a_{v_i}|\A^{u_i;v_i}|.
\end{equation}
Expanding $|\A|$ in the $u_i$-th row gives, using (\ref{bAiu}), 
\begin{align*}
|\A|&=\left(\sum_{j\not\in\{u_i,v_i\}}(-1)^{u_i+j}A_{u_i,j}|\A^{u_i;j}|\right)
+a_{u_i}|\A^{u_i;u_i}| + (-1)^{u_i+v_i} b_i |\A^{u_i;v_i}|\\
&=(a_{u_i}-c_i)|\A^{u_i;u_i}|+ (-1)^{u_i+v_i} (b_i - a_{v_i})|\A^{u_i;v_i}|.
\end{align*}
The proof of the case $w=v_i$ is similar. Let $\bar{\A}$ be the matrix obtained from $\A$ by replacing the $v_i$-th row by the $u_i$-th row.
By expanding in this row, we find, as $\A_{u_i,j}=\A_{v_i,j}$ when $j\not\in\{u_i,v_i\}$,
\begin{equation}\label{bAiv}
0=|\bar{\A}|=\left(\sum_{j\not\in\{u_i,v_i\}}(-1)^{v_i+j}A_{v_i,j}|\A^{v_i;j}|\right) + (-1)^{u_i+v_i}a_{u_i}|\A^{v_i;u_i}|+b_i|\A^{v_i;v_i}|.
\end{equation}
Expanding $|\A|$ in the $v_i$-th row gives, using (\ref{bAiv}), 
\begin{align*}
|\A|&=\left(\sum_{j\not\in\{u_i,v_i\}}(-1)^{v_i+j}A_{v_i,j}|\A^{v_i;j}|\right) + (-1)^{u_i+v_i}c_i|\A^{v_i;u_i}|+a_{v_i}|\A^{v_i;v_i}|\\
&=(-1)^{u_i+v_i}(c_i-a_{u_i})|\A^{v_i;u_i}|+(a_{v_i}-b_i)|\A^{v_i;v_i}|.
\end{align*}
\end{proof}

The proof of Theorem \ref{zhro} also relies on a special case of Sylvester's determinental identity, i.e. that, for any Matrix $\A$, with $i<j,k<l$,
\begin{equation} \label{Syl}
|\A^{i;k}||\A^{j;l}|-|\A^{i;l}||\A^{j;k}|
=|\A||\A^{i,j;k,l}|.
\end{equation}
We note that \eqref{Syl} also holds if $i>j,k>l$, and that, for $i<j,k>l$ or $i>j,k<l$,
\begin{equation} \label{Syli}
|\A^{i;k}||\A^{j;l}|-|\A^{i;l}||\A^{j;k}|
=-|\A||\A^{i,j;k,l}|.
\end{equation}
Of course, when $i=j$ or $k=l$ the expression $|\A^{i;k}||\A^{j;l}|-|\A^{i;l}||\A^{j;k}|$ is identically 0.

\begin{proof}[Proof of Theorem \ref{zhro}.]
The matrix $\Z$ has rank 1 if each row is a multiple of the first one. We will prove that
\[
|\A|=0 \Rightarrow \frac{Z_{i,k}}{Z_{1,k}}=\frac{Z_{i,1}}{Z_{1,1}},
\]
by showing that $|\A|$ is a divisor of $Y_{i,k}:=Z_{i,k}Z_{1,1}-Z_{i,1}Z_{1,k}$. 
To evaluate $Y_{i,k}$ we will use Proposition \ref{FLCM}. Hence we need to distinguish cases (for example, to evaluate $Z_{i,1}$ we need to know that $u_i\neq 1$ or $v_i\neq 1$), as in Figure \ref{cases}.

\begin{figure}[h]
\begin{center}
\scalebox{.9}{\begin{tikzpicture}
    \node (1) at (0,5) {$\bullet$};
	\node (2) at (0,4) {$v_1\neq 1$};
	\node (3) at (0,3) {$v_i\neq 1$};
	\node (4) at (0,2) {$v_i\neq k$};
	\node (5) at (0,1) {$v_1\neq k$};
    \node (6) at (2,5) {$v_1=1$};
	\node (7) at (2,4) {$v_i=1$};
	\node (8) at (2,3) {$v_i=k$};
	\node (9) at (2,2) {$v_1=k$};
	\node (10) at (2,1) {$u_1\neq k$};
	\node (11) at (4,5) {$u_1\neq 1$};
	\node (12) at (4,4) {$u_i\neq 1$};
	\node (13) at (4,3) {$u_i\neq k$};
	\node (14) at (4,2) {$v_1\neq k$};
	\node (15) at (6,5) {$v_1\neq k$};
	\node (16) at (6,4) {$v_i\neq k$};
	\node (17) at (6,3) {$v_1=k$};
	\node (18) at (6,2) {$u_1\neq k$};
    \node (19) at (6,6) {$u_i\neq 1$};
    \node (20) at (6,7) {$v_i\neq k$};
    \node (21) at (8,6) {$v_i=k$};
    \node (22) at (10,6) {$u_i\neq k$};
    \node (23) at (8,5) {$u_i=1$};
    \node (24) at (10,5) {$v_i\neq 1$};
    \node (25) at (12,5) {$u_i\neq k$};
    \node (26) at (8,4) {$v_1=k$};
    \node (27) at (10,4) {$u_1\neq k$};
    \node (28) at (8,3) {$v_1\neq k$};
	\draw[->] (1) -- (2);
    \draw[->] (2) -- (3);
    \draw[->] (3) -- (4);
    \draw[->] (4) -- (5);
    \draw[->] (1) -- (6);
    \draw[->] (2) -- (7);
    \draw[->] (3) -- (8);
    \draw[->] (4) -- (9);
    \draw[->] (9) -- (10);
    \draw[->] (6) -- (11);
    \draw[->] (7) -- (12);
    \draw[->] (8) -- (13);
    \draw[->] (13) -- (14);
    \draw[->] (11) -- (15);
    \draw[->] (12) -- (16);
    \draw[->] (13) -- (17);
    \draw[->] (17) -- (18);
    \draw[->] (15) -- (19);
    \draw[->] (19) -- (20);
    \draw[->] (19) -- (21);
    \draw[->] (21) -- (22);
    \draw[->] (15) -- (23);    
    \draw[->] (23) -- (24);
    \draw[->] (24) -- (25);
    \draw[->] (16) -- (26);
    \draw[->] (26) -- (27);
    \draw[->] (16) -- (28);
    \node[shape=circle,draw=black] at (0,0) {$i$};
    \node[shape=circle,draw=black] at (2,0) {$ii$};      
    \node[shape=circle,draw=black] at (4,1) {$iii$};      
    \node[shape=circle,draw=black] at (6,1) {$iv$};      
    \node[shape=circle,draw=black] at (8,2) {$v$};      
    \node[shape=circle,draw=black] at (12,4) {$vi$};      
    \node[shape=circle,draw=black] at (14,5) {$vii$};      
    \node[shape=circle,draw=black] at (12,6) {$viii$};            
    \node[shape=circle,draw=black] at (8,7) {$ix$};      
    \end{tikzpicture}}
\caption{\label{cases} Nine cases to analyse.}
\end{center}
\end{figure}
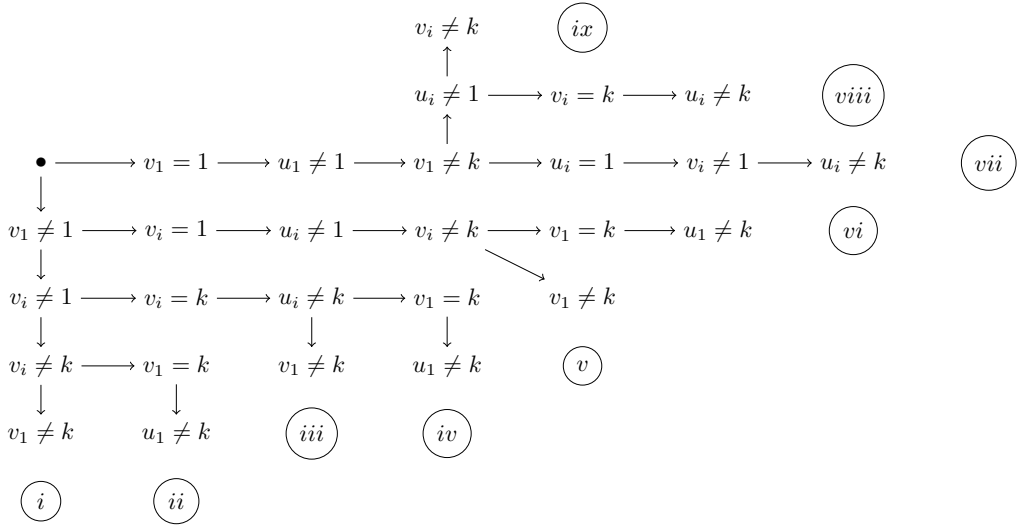
\begin{itemize}
\item Case $i$: $v_1\neq 1,v_i\neq 1,v_i\neq k,v_1\neq k$. Using Proposition \ref{FLCM}, and the identities (\ref{Syl}), (\ref{Syli}), we find, for $u_1\neq u_i$,
    \begin{align*}
    Y_{i,k}&=(-1)^{1+k+u_1+u_i}(c_1-a_{u_1})(c_i-a_{u_i})(|\A^{k;u_i}||\A^{1;u_1}|-|\A^{1;u_i}||\A^{k;u_1}|)\\
    &=(-1)^{1+k+u_1+u_i}(c_1-a_{u_1})(c_i-a_{u_i})(\pm|\A||\A^{1,k;u_1,u_i}|).
    \end{align*}
    If $u_1=u_i$ then $Y_{i,k}$ vanishes. 
\item Case $ii$: $v_1\neq 1,v_i\neq 1,v_i\neq k,u_1\neq k$. We find
    \[
    Y_{i,k}=(c_i-a_{u_i})\left((-1)^{1+k+u_1+u_i}((c_1-a_{u_1})|\A^{k;u_i}||\A^{1;u_1}|-(-1)^{1+k+v_1+u_i}(b_1-a_{v_1})|\A^{1;u_i}||\A^{k;v_1}|\right).
    \]
    When $u_1\neq 1$ we can use Proposition \ref{ROM} to show that this equals
    \[
    Y_{i,k}=(-1)^{1+k+v_1+u_i}(b_1-a_{v_1})(c_i-a_{u_i})(|\A^{k;u_i}||\A^{1;v_1}|-|\A^{1;u_i}||\A^{k;v_1}|).
    \]
    which has $|\A|$ as divisor. When $u_1=1$ we use Proposition \ref{FAE} to show that modulo $|\A|$ we get the same expression
    \[
    Y_{i,k}\equiv(-1)^{1+k+v_1+u_i}(b_1-a_{v_1})(c_i-a_{u_i})(|\A^{k;u_i}||\A^{1;v_1}|-|\A^{1;u_i}||\A^{k;v_1}|).
    \]
\item Cases $iii,\ldots,ix$. All remaining cases can be proven similar to the above. E.g., for case $iii$, $v_1\neq 1,v_i\neq 1,u_i\neq k,v_1\neq k$, we find
    \begin{align*}
    Y_{i,k}&=(c_1-a_{u_1})\left((-1)^{1+k+u_1+v_i}(b_i-a_{v_i})|\A^{1;u_1}||\A^{k;v_i}|\right.\\
    & \ \ \ \left.-(-1)^{1+k+u_1+u_i}(c_i-a_{u_i})|\A^{k;u_1}||\A^{1;u_i}|\right)\\
    &\equiv (-1)^{1+k+u_1+v_i}(c_1-a_{u_1})(b_i-a_{v_i})(|\A^{k;v_i}||\A^{1;u_1}|-|\A^{1;v_i}||\A^{k;u_1}|).
    \end{align*}
    For the cases $ix,xi,xii$ and $ix$ one needs to use Propositions \ref{ROM} and/or \ref{FAE} twice.
\end{itemize}
\end{proof}

\noindent
{\bf Competing Interests and Data availability statements}

\noindent
On behalf of all authors, the corresponding author states that there is no conflict of interest and data sharing is not applicable to this article as no datasets were generated or analysed during the current study.

\end{document}